\documentclass[smallcondensed]{svjour3}

\smartqed
\usepackage{amsmath}
\usepackage{amssymb}
\usepackage{bbold}
\usepackage{epsfig}
\usepackage{makeidx}  

\usepackage{mathrsfs}

\title{Non Abelian Bent Functions}

\author{Laurent Poinsot}

\institute{LIPN-UMR CNRS 7030, Institut Galil\'ee, Universit\'e Paris-Nord 
99, avenue Jean-Baptiste Cl\'ement  
93430 Villetaneuse, France, \\\email{laurent.poinsot@lipn.univ-paris13.fr}}

\begin{document}

\maketitle

\begin{abstract}
Perfect nonlinear functions from a finite group $G$ to another one $H$ are those functions $f: G \rightarrow H$ such that for all nonzero $\alpha \in G$, the derivative $d_{\alpha}f: x \mapsto f(\alpha x) f(x)^{-1}$ is balanced. In the case where both $G$ and $H$ are Abelian groups, $f: G \rightarrow H$ is perfect nonlinear if and only if $f$ is bent {\it i.e.} for all nonprincipal character $\chi$ of $H$, the (discrete) Fourier transform of $\chi \circ f$  has a constant magnitude equals to $|G|$. In this paper, using the theory of linear representations, we exhibit similar bentness-like characterizations in the cases where $G$ and/or $H$ are (finite) non Abelian groups. Thus we extend the concept of bent functions to the framework of non Abelian groups.
\end{abstract}

\section{Introduction}
Let $G$ and $H$ be two finite groups (in multiplicative representation). Perfect nonlinear functions from $G$ to $H$ are those ideal functions $f: G \rightarrow H$ that match the less possible with the pattern of group homomorphism {\it i.e.} such that for all nonzero $\alpha \in G$ and for all $\beta \in H$,
\begin{equation}
|\{x \in G | f(\alpha x) f(x)^{-1} = \beta \}| = \displaystyle \frac{|G|}{|H|}.
\end{equation}
When $G$ and $H$ are both (finite-dimensional) vector spaces over ${\mathbb{Z}}_2 = \{0,1\}$, these functions, originally introduced by Nyberg \cite{Nyb92}, exhibit the maximal resistance against the differential attack \cite{Biha1}. Also in the Boolean case, this notion is known to be equivalent to bent functions: $f: {\mathbb{Z}}_2^m \rightarrow {\mathbb{Z}}_2^n$ is bent if for all $\alpha \in {\mathbb{Z}}_2^m$ and for all nonzero $\beta$ in ${\mathbb{Z}}_2^n$, 
\begin{equation}
|\widehat{(\chi^{\beta}_{{\mathbb{Z}}_2^n}  \circ f)}(\alpha)|^2 = 2^m
\end{equation}
where $\chi^{\beta}_{{\mathbb{Z}}_2^n}: {\mathbb{Z}}_2^n \rightarrow \{\pm 1\}$ is defined at $y$ by $(-1)^{\beta.y}$ (the point in exponent is the natural dot-product of ${\mathbb{Z}}_2^n$) and 
\begin{equation}
\widehat{\phi}(\alpha) = \displaystyle \sum_{x \in {\mathbb{Z}}_2^m}\phi(x)(-1)^{\alpha.x}
\end{equation}
is the Fourier transform of $\phi: G \rightarrow {\mathbb{C}}$ (this time $\alpha.x$ is the dot-product in ${\mathbb{Z}}_2^m$). These functions, independently introduced by Dillon \cite{Dill} and Rothaus \cite{Rot76}, exhibit the maximal resistance against the linear attack \cite{Mat94}.\\
The equivalence between bentness and perfect nonlinearity has been recently extended by Carlet and Ding \cite{CD04} and Pott \cite{Pot04} to the general case: $f: G \rightarrow H$ (where $G$ and $H$ are two finite Abelian groups) is perfect nonlinear if and only if for all $\alpha \in G$ and for all nonprincipal character $\chi$ of $H$, 
\begin{equation}
|\widehat{(\chi \circ f)}(\alpha)|^2 = |G|
\end{equation}
where for $\phi: G \rightarrow {\mathbb{C}}$, $\widehat{\phi}$ is its discrete Fourier transform.\\
In this paper, we exhibit the same kind of characterizations in the case where at least one of the two finite groups $G$ and $H$ is non Abelian. This gives a general equivalence between perfect nonlinearity and bentess.
\subsection*{Outline of the paper}
Next section contains the general notations used in the paper. In sect. \ref{sect2} are recalled some of the main results on the duality of finite groups. In particular, we present several kind of Fourier transforms used in the new characterizations of perfect nonlinearity in the non Abelian cases. The notion of perfect nonlinearity is exposed in the general framework of finite groups in sect. \ref{sect3}. Also in this section is given the dual characterization - {\it i.e.} the notion of bentness - of Carlet, Ding \cite{CD04} and Pott \cite{Pot04} of perfect nonlinear functions in the Abelian groups setting. Finally our own results of non Abelian bentness are developed in sect. \ref{sect4}.

\section{Notations}\label{sect1}
$|S|$ is the cardinality of any finite set $S$ and if $S$ is nonempty (possibly infinite), ${\mathit{Id}}_S$ denotes its identity map.\\
In this paper, the capital letters ``$G$'' and ``$H$'' always denote finite groups in multiplicative representation, $e_G$ is the neutral element and $G^*$ is defined as $G \setminus \{e_G\}$.\\
The vector spaces considered are always finite-dimensional complex vector spaces. For a (complex) vector space $V$, $0_V$ is its zero, $\dim_{{\mathbb{C}}}(V)$ is its dimension and ${\mathit{GL}}(V)$ denotes the {\it linear group} of $V$ which is a subset of the vector space of all endomorphisms of $V$, denoted ${\mathit{End}}(V)$. If $\lambda$ is any linear map, $\lambda^*$ denotes its adjoint and the {\it unitary group} of $V$ is ${\mathbb{U}}(V)$.\\

By convention, for each known result recalled in this paper, the proof has been intentionally omitted and a reference - not necessarily {\it the} original reference - is given. On the contrary, our own results are given obviously with their proofs and without any reference.

\section{On the duality of finite groups}\label{sect2}

\subsection{Introduction}
This paper is dedicated to the establishment of a dual characterization of the concept of perfect nonlinearity - similar to the one given by Carlet and Ding \cite{CD04} and Pott \cite{Pot04} concerning Abelian groups - in the non Abelian groups setting, using some harmonic analysis techniques. So in this section, we recall some basics about the duality of finite groups and the Fourier transform. Most of the definitions and results given in this section are well-known and can be found in any book on finite groups (\cite{pey04} for instance). 

\subsection{The Abelian case}

\subsubsection{The theory of characters}

\begin{definition}
\textup{Let $G$ be a finite group. A {\it character} $\chi$ of $G$ is a group homomorphism from $G$ to the multiplicative group ${\mathbb{C}}^*$. We denote by $\widehat{G}$ the set of all characters, called {\it dual} of $G$.}
\end{definition}
The set $\widehat{G}$ is actually a group under the point-wise multiplication of characters and its elements are valued in the group of the $|G|^{\mathit{th}}$ complex roots of the unity. In particular,
\begin{equation}
\forall \chi \in \widehat{G},\ \forall x \in G,\ |\chi(x)| = 1\ \mbox{and}\ \chi(x^{-1}) = \overline{\chi(x)}
\end{equation}
where $|z|$ is the complex-modulus and $\overline{z}$ is the conjugate of $z \in {\mathbb{C}}$.

\begin{proposition}\label{prop_G_Gchap_isom}(\cite{pey04})
Let $G$ be a finite Abelian group. Then $G$ and $\widehat{G}$ are isomorphic.
\end{proposition}
In the remainder of this paper, when a finite Abelian group $G$ is considered, we always implicitly suppose that an isomorphism from $G$ to $G^*$ has been fixed and we use $\chi_G^{\alpha}$ to denote the image of $\alpha$ by such an isomorphism. In particular, $\forall x \in G$, $\chi^{e_G}_G (x) = 1$ (this character is called {\it trivial} or {\it principal}). Finally the characters satisfy the well-known orthogonality properties.
\begin{lemma}\label{lemme_somme_caractere}(\cite{pey04})
Let $G$ be a finite Abelian group. For all $\alpha \in G$ we have
\begin{equation}
\displaystyle \sum_{x \in G} \chi_G^{\alpha}(x) = \left\{
\begin{array}{l l}
0 & \mbox{if}\ \alpha \in G^*,\\
|G| & \mbox{if}\ \alpha = e_G.
\end{array}
\right.
\end{equation}
For all $x \in G$, we have
\begin{equation}
\displaystyle \sum_{\alpha \in G} \chi_G^{\alpha}(x) = \left\{
\begin{array}{l l}
0 & \mbox{if}\ x \in G^*,\\
|G| & \mbox{if}\ x = e_G.
\end{array}
\right.
\end{equation}
\end{lemma}
\subsubsection{The (discrete) Fourier transform}

\begin{definition}
\textup{Let $G$ be a finite Abelian group. The {\it Fourier transform} of $\phi: G \rightarrow {\mathbb{C}}$ is the map $\widehat{\phi}: G \rightarrow {\mathbb{C}}$ defined for $\alpha \in G$ by
\begin{equation}
\displaystyle \widehat{\phi}(\alpha) = \sum_{x \in G} \phi(x)\chi^{\alpha}_G (x).
\end{equation}}
\end{definition}
Some well-known and useful results are summarized below for $G$ a finite Abelian group.

\begin{proposition}\label{prop_TF_discrete}(\cite{pey04})
Let $\phi: G \rightarrow {\mathbb{C}}$.
\begin{enumerate}
\item We have the {\it inversion formula}
\begin{equation}
\phi = \displaystyle \frac{1}{|G|}\sum_{\alpha \in G}\widehat{\phi}(\alpha) \overline{\chi^{\alpha}_G};
\end{equation}
\item We have the {\it Parseval's equation}
\begin{equation}
\displaystyle \sum_{x \in G}|\phi(x)|^2 = \frac{1}{|G|}\sum_{\alpha \in G}|\widehat{\phi}(\alpha)|^2;
\end{equation}
\item $\phi (x) = 0$ $\forall x \in G^*$ if and only if $\widehat{\phi}$ is constant;
\item $\widehat{\phi}(\alpha) = 0$ $\forall \alpha \in G^*$ if and only if $\phi$ is constant.
\end{enumerate}
\end{proposition}
Concerning the last two points, their proofs can be checked in \cite{CD04}.

\subsubsection{The multidimensional (discrete) Fourier transform}
In subsection B. of sect. \ref{sect4}, we consider some $V$-valued functions, where $V$ is a finite-dimensional vector space, defined on a finite Abelian group $G$ and we need to compute their ``Fourier transforms''. That is the reason why we now introduce a natural extension of the discrete Fourier transform, called {\it multidimensional Fourier transform}, to deal with $V$-valued functions rather than ${\mathbb{C}}$-valued ones. More details on this transform should be found in \cite{Poi05}.
\begin{definition}
\textup{Let $G$ be a finite Abelian group, $V$ a finite-dimensional vector space over ${\mathbb{C}}$ and $\phi: G \rightarrow {\mathbb{C}}$. The {\it multidimensional Fourier transform} of $\phi$ is defined as
\begin{equation}
\begin{array}{l l l l}
\widehat{\phi}^{\mathit{MD}}: & G & \rightarrow & V\\
& \alpha & \mapsto & \displaystyle \sum_{x \in G}\chi^{\alpha}_G (x) \phi(x).
\end{array}
\end{equation}}
\end{definition}
Now let suppose that $V$ is equipped with an inner-product (linear in the first variable and anti-linear in the second) denoted $\langle ., . \rangle_V$. Let fix $B$ an orthonormal basis of $V$. For each $e \in B$, we define the {\it component function} $\phi_e$ of $\phi: G \rightarrow V$ in direction $e$ as the map
\begin{equation}
\begin{array}{l l l l}
\phi_e: & G & \rightarrow & {\mathbb{C}}\\
& x & \mapsto & \langle \phi(x),e\rangle_V.
\end{array}
\end{equation}
According to the properties of orthonormal basis, we have $\forall x \in G$,
\begin{equation}
\phi (x) = \displaystyle \sum_{e \in B} \phi_e(x) e.
\end{equation}
We can easily check that the multidimensional Fourier transform is actually a component-wise discrete Fourier transform given by the following equation (for $\alpha \in G$)
\begin{equation}
\widehat{\phi}^{\mathit{MD}}(\alpha) = \displaystyle \sum_{e \in B} \widehat{\phi_e}(\alpha) e.
\end{equation}
Using this last equation, we can establish an {\it inversion formula} for the multidimensional Fourier transform. Let $x \in X$.
\begin{equation}\label{inv_form_multi_dim_FT}
\begin{array}{l l l}
\phi (x) &=& \displaystyle \sum_{e \in B} \phi_e (x) e\\
&=& \displaystyle \sum_{e \in E} \left( \frac{1}{|G|} \sum_{\alpha \in G} \widehat{\phi_e}(\alpha) \overline{\chi^{\alpha}_G}(x)\right)e\\
&&\mbox{(by the inversion formula applied on $\phi_e$)}\\
&=&\displaystyle \frac{1}{|G|} \sum_{\alpha \in G}\overline{\chi^{\alpha}_G (x)} \left( \sum_{e \in B} \widehat{\phi_e}(\alpha) e\right)\\
&=& \displaystyle \frac{1}{|G|} \sum_{\alpha \in G} \overline{\chi^{\alpha}_G (x)} \widehat{\phi}^{\mathit{MD}}(\alpha).
\end{array}
\end{equation}
Finally this transform satisfies a result similar to the third point of proposition \ref{prop_TF_discrete}.
\begin{proposition}\label{prop_TF_multidimensional}
Let $G$ be a finite Abelian group, $V$ a finite-dimensional vector space over ${\mathbb{C}}$ and $\phi: G \rightarrow V$. We have $\phi(x) = 0_{V}$ for all $x \in G^*$ if and only if $\widehat{\phi}^{\mathit{MD}}(\alpha) = \phi (e_G)$ for all $\alpha \in G$. 
\end{proposition}

\begin{proof}
\begin{itemize}
\item [$\Rightarrow$)] Let suppose that for all $x \in G^*$, $\phi(x) = 0_{V}$. Then we have $\forall \alpha \in G$, $\widehat{\phi}^{\mathit{MD}}(\alpha) = \displaystyle \sum_{x \in G} \chi^{\alpha}_G (x) \phi(x) = \phi(e_G)$. 
\item [$\Leftarrow$)] Let suppose that $\widehat{\phi}^{\mathit{MD}}(\alpha) = \phi (e_G)$ for all $\alpha \in G$. Using the inversion formula (\ref{inv_form_multi_dim_FT}), we get that for $x \in G$, $\phi(x) = \displaystyle \frac{1}{|G|}\left(\sum_{\alpha \in G} \overline{\chi^{\alpha}_G (x)}\right) \phi(e_G) = \left \{
\begin{array}{l l}
0_{V} & \mbox{if}\ x \in G^*,\\
\phi(e_G) & \mbox{if}\ x = e_G,
\end{array}\right .$ according to the orthogonality relations (lemma \ref{lemme_somme_caractere}). 
\end{itemize}
\qed
\end{proof}

\subsection{The non Abelian case}

\subsubsection{The theory of linear representations}

\begin{definition}
\textup{Let $V$ be a finite-dimensional complex vector space. A {\it linear representation} of a finite group $G$ on $V$ is a group homomorphism from $G$ to ${\mathit{GL}}(V)$.}
\end{definition}
For each linear representation $\rho: G \rightarrow {\mathit{GL}}(V)$, it is possible to find a basis of $V$ in which for all $x \in G$, $\rho(x)$ is a unitary operator of $V$ {\it i.e.} $\rho: G \rightarrow {\mathbb{U}}(V)$. Indeed, we can check that for a linear representation $\rho$ of $G$ on $V$, for each $x \in G$, $\rho(x)$ leaves invariant the following inner-product in $V$
\begin{equation}\label{inner_product_inv}
\langle u,v\rangle_{G,\rho,V} = \displaystyle \sum_{x \in G} \langle \rho(x)(u),\rho(x)(v)\rangle_V
\end{equation}
where $(u,v) \in V^2$ and $\langle .,.\rangle_V$ denotes any inner-product of $V$ (linear in the first variable and anti-linear in the second). Then in the remainder, without loss of generality, we only consider unitary representations.\\

The linear representations of $G$ on ${\mathbb{C}}$ can be identified with the characters of $G$ since ${\mathbb{C}}$ and ${\mathit{GL}}({\mathbb{C}})$ are isomorphic. Actually if $G$ is a finite Abelian group then the notion of linear representation gives nothing new because it is equivalent to the notion of character.

\begin{definition}
\textup{A linear representation $\rho$ of a finite group $G$ on $V$ is said to be {\it irreducible} if there is no subspace $W \subset V$, other than $\{0_V\}$ and $V$, such that $\forall x \in G$, $\forall w \in W$, $\rho(x)(w) \in W$.}
\end{definition}
\begin{definition}
\textup{Two linear representations $\rho$ and $\rho'$ of a finite group $G$ on respectively $V$ and $V'$ are {\it isomorphic} if it exists a linear isomorphism $\Phi: V \rightarrow V'$ such that for all $x \in G$,
\begin{equation}\label{equivariance}
\Phi \circ \rho (x) = \rho' (x) \circ \Phi.
\end{equation}}
(A linear map that satisfies equality~(\ref{equivariance}) is called \emph{equivariant map} and is easily seen to be a morphism from $V$ to $V'$ seen as left $G$-modules.)
\end{definition}
The notion of isomorphism is an equivalence relation for linear representations\footnote{Even if the collection of all linear representations of a given finite group does not form a set but rather a proper class, it is an easy exercice to check that the collection of isomorphism classes of linear representations really does form a set.}.
\begin{definition}
\textup{For a finite group $G$, the {\it dual} of $G$, denoted $\widetilde{G}$, is a set that contains exactly one and only one representative of each equivalence class of isomorphic irreducible representations of $G$.} 
\end{definition}
By definition, if $(\rho,\rho') \in \widetilde{G}^2$, then $\rho$ and $\rho'$ are nonisomorphic irreducible representations of $G$. In the remainder, the notation
\begin{equation}
\rho = \rho_V \in \widetilde{G}
\end{equation}
means that $\rho: G \rightarrow {\mathbb{U}}(V)$ is an irreducible representation of $G$.\\
If $G$ is a finite Abelian group, then $\widetilde{G}$ is equal to $\widehat{G}$ (up to an isomorphism from ${\mathit{GL}}({\mathbb{C}})$ to ${\mathbb{C}}$). If $G$ is a finite non Abelian group, the two notions of duality become distinct (in particular, $\widetilde{G}$ is not a group). By abuse of notation, $\widetilde{G}^*$ is defined as the set $\widetilde{G} \setminus \{\rho_0\}$ where $\rho_0$ is the {\it trivial} or {\it principal} representation of $G$ {\it i.e.} $\forall x \in G$, $\rho_0(x) = {\mathit{Id}}_{{\mathbb{C}}}$.\\

When dealing with linear representations, a major result, know as {\it Schur's lemma}, should be kept in mind.
\begin{lemma}\label{schur_lemma}(\cite{pey04})
Let $G$ be a finite group. Let $\rho = \rho_V \in \widetilde{G}$ and $\lambda \in {\mathit{End}}(V)$. If $\forall x \in G$, $\lambda \circ \rho(x) = \rho(x) \circ \lambda$ then $\lambda$ is a multiple of the identity {\it i.e.} it exists $k \in {\mathbb{C}}$ such that $\lambda = k {\mathit{Id}}_V$.
\end{lemma}
As direct consequences of the Schur's lemma, we can state the two following results that will be use in the sequel.
\begin{lemma}\label{petit_theorem_important_sur_les_caracteres}(\cite{pey04})
Let $G$ be a finite group. For $x \in G^*$, 
\begin{equation}
\displaystyle \sum_{\rho_V \in \widetilde{G}} \dim_{{\mathbb{C}}}(V) {\mathit{tr}}(\rho_V (x)) = 0
\end{equation}
where ``${\mathit{tr}}$'' denotes the usual trace of endomorphisms.
\end{lemma}

\begin{lemma}\label{lemma_star}
Let $G$ be a finite group. Let $\rho = \rho_V \in \widetilde{G}^*$. Then
\begin{equation}
\displaystyle \sum_{x \in G}\rho(x) = 0_{{\mathit{End}}(V)}.
\end{equation}
\end{lemma}

\begin{proof}
Let $\lambda \in {\mathit{End}}(V)$ defined as $\lambda = \displaystyle \sum_{x \in G} \rho(x)$. Let $x_0 \in G$. We have
\begin{equation*}
\lambda = \displaystyle \sum_{x \in G}\rho(x) = \sum_{x \in G} \rho(x_0 x) = \rho(x_0) \circ \sum_{x \in G} \rho(x) = \rho(x_0) \circ \lambda
\end{equation*}
but also\\
$\lambda = \displaystyle \sum_{x \in G}\rho(x) = \sum_{x \in G} \rho(x x_0) = \left(\sum_{x \in G} \rho(x)\right) \circ \rho(x_0) = \lambda \circ \rho(x_0).$\\
In particular $\lambda \circ \rho(x_0) = \rho(x_0) \circ \lambda$. As it is true for any $x_0 \in G$, $\lambda$ commutes with all $\rho(x)$. By the Schur's lemma, $\lambda$ is a multiple on the identity: it exists $k \in {\mathbb{C}}$ such that $\lambda = k {\mathit{Id}}_{V}$. Now let suppose $\lambda \not = 0_{{\mathit{End}}(V)}$, then $k \in {\mathbb{C}}^*$. Using the first part of the proof, we know that $\lambda = \rho(x) \circ \lambda$ (for each $x \in G$). Then $({\mathit{Id}}_V - \rho(x)) \circ \lambda = 0_{{\mathit{End}}(V)}$. As $\lambda = k {\mathit{Id}}_V$, we have $({\mathit{Id}}_V - \rho(x)) \circ (k {\mathit{Id}}_V) = 0_{{\mathit{End}}(V)}$. Since $k \not = 0$, we have ${\mathit{Id}}_V - \rho(x) = 0_{{\mathit{End}}(V)}$ or also $\rho(x) = {\mathit{Id}}_V$ which is a contradiction with the assumption that $\rho$ is non trivial. \qed
\end{proof}

\subsubsection{The representation-based Fourier transform}
By substituting irreducible linear representations to characters, it is possible to define a kind of Fourier transform for non Abelian groups.\\
Let $G$ be any finite group. 
\begin{definition}
\textup{Let $\phi: G \rightarrow {\mathbb{C}}$. The {\it (representation-based) Fourier transform} of $\phi$ is defined for $\rho = \rho_V \in \widetilde{G}$ as
\begin{equation}
\widetilde{\phi}(\rho) = \displaystyle \sum_{x \in G} \phi(x) \rho(x) \in {\mathit{End}}(V).
\end{equation}}
\end{definition}
This notion is a generalization of the classical discrete Fourier transform. This transform is invertible so we have also an {\it inversion formula}.
\begin{proposition}\label{prop_somme_de_trace_nulle}(\cite{pey04})
Let $\phi: G \rightarrow {\mathbb{C}}$. Then for all $x \in G$ we have,
\begin{equation}
\phi(x) = \displaystyle \frac{1}{|G|}\sum_{\rho_V \in \widetilde{G}} \dim_{{\mathbb{C}}}(V) {\mathit{tr}}(\rho_V (x^{-1}) \circ \widetilde{\phi}(\rho_V)).
\end{equation}
\end{proposition}
A last technical lemma is given below.
\begin{lemma}\label{lemma_diese}
Let $\phi: G \rightarrow {\mathbb{C}}$. We have
\begin{enumerate}
\item $\phi(x) = 0$ $\forall x \in G^*$ if and only if $\forall \rho = \rho_V \in \widetilde{G}$, $\widetilde{\phi}(\rho) = \phi(e_G) {\mathit{Id}}_V$;
\item $\widetilde{\phi}(\rho) = 0_{{\mathit{End}}(V)}$ $\forall \rho = \rho_V \in \widetilde{G}^*$ if and only if $\phi$ is constant.
\end{enumerate}
\end{lemma}

\begin{proof}

\begin{enumerate}
\item \begin{itemize}
\item[$\Rightarrow$)] For $\rho = \rho_V \in \widetilde{G}$, we have
\begin{equation*}
\begin{array}{l l l}
\widetilde{\phi}(\rho) & = & \displaystyle \sum_{x \in G}\phi(x)\rho (x)\ \mbox{(by definition)}\\
& = & \phi(e_G)\rho(e_G)\ \mbox{(by assumption on $\phi$)}\\
& = & \phi(e_G){\mathit{Id}}_{V}\\
&&\mbox{(since $\rho$ is a group homomorphism)}.
\end{array}
\end{equation*}
\item[$\Leftarrow$)] For $x \in G$, the inversion formula gives
\begin{equation*}
\begin{array}{l}
\phi(x)  = \displaystyle \frac{1}{|G|}\sum_{\rho_V \in \widetilde{G}}\dim_{{\mathbb{C}}}(V) {\mathit{tr}}(\rho_V (x^{-1}) \circ \widetilde{\phi} (\rho_V))\\
 = \displaystyle \frac{1}{|G|}\sum_{\rho_V \in \widetilde{G}}\dim_{{\mathbb{C}}}(V) {\mathit{tr}}(\rho_V (x^{-1}) \circ \phi(e_G){\mathit{Id}}_{V})\\
\mbox{(by hypothesis)}\\
 = \displaystyle \frac{\phi(e_G)}{|G|}\sum_{\rho_V \in \widetilde{G}}\dim_{{\mathbb{C}}}(V) {\mathit{tr}}(\rho_V (x^{-1}))\\
 = \displaystyle \frac{\phi(e_G)}{|G|}\sum_{\rho_V \in \widetilde{G}}\dim_{{\mathbb{C}}}(V) {\mathit{tr}}(\rho_V (x)^{-1})\\
\mbox{(since $\rho_V$ is a group homomorphism)}\\
 = \displaystyle \frac{\phi(e_G)}{|G|}\sum_{\rho_V \in \widetilde{G}}\dim_{{\mathbb{C}}}(V) {\mathit{tr}}(\rho_V (x)^{*})\\
\mbox{(since $\rho_V(x)$ is unitary)}\\
 = \displaystyle \frac{\phi(e_G)}{|G|}\sum_{\rho_V \in \widetilde{G}}\dim_{{\mathbb{C}}}(V) \overline{{\mathit{tr}}(\rho_V (x))}\\
 = \displaystyle \frac{\phi(e_G)}{|G|}\overline{\sum_{\rho_V \in \widetilde{G}}\dim_{{\mathbb{C}}}(V) {\mathit{tr}}(\rho_V (x))}\\
 = 0\ \mbox{if $x \not = e_G$ (according to lemma \ref{petit_theorem_important_sur_les_caracteres})}.
\end{array}
\end{equation*}
\end{itemize}
\item \begin{itemize}
\item[$\Rightarrow$)] By the inversion formula, $\forall x \in G$,
\begin{equation}
\begin{array}{l}
\phi(x) =\displaystyle \frac{1}{|G|}\sum_{\rho_V \in \widetilde{G}} \dim_{{\mathbb{C}}}(V) {\mathit{tr}}(\rho_V (x^{-1}) \circ \widetilde{\phi}(\rho_V))\\
= \displaystyle \frac{1}{|G|}{\mathit{tr}}(\widetilde{\phi} ({\mathit{Id}}_{{\mathbb{C}}}))\ \mbox{(by hypothesis)}.
\end{array}
\end{equation}

\item[$\Leftarrow$)] Let $\rho_V \in \widetilde{G}$, we have $\widetilde{\phi}(\rho_V) = k \displaystyle \sum_{x \in G}\rho_V (x)$ (with $\phi(x) = k\ \forall x \in G$). According to lemma \ref{lemma_star}, we deduce that $\widetilde{\phi}(\rho_V) = 0_{{\mathit{End}}(V)}$ for all $\rho_V \in \widetilde{G}^*$.
\end{itemize}
\end{enumerate} \qed
\end{proof}

\section{On perfect nonlinear functions}\label{sect3}

\subsection{Some basic definitions}
Perfect nonlinearity must be seen as the fundamental notion on which our results are based. Actually our ambition in this paper is to describe this combinatorial concept in terms of Fourier transforms. Thus it is necessary to briefly present this topic.
\begin{definition}
\textup{Let $X$ and $Y$ be two finite nonempty sets. A function $f: X \rightarrow Y$ is said to be {\it balanced} if the function 
\begin{equation}
\begin{array}{l l l l}
\phi_f:& Y  & \rightarrow & \mathbb{N}\\
& y & \mapsto & |\{x \in X | f(x) = y\}|
\end{array}
\end{equation}
is constant equal to $\displaystyle \frac{|X|}{|Y|}$.}
\end{definition}

\begin{definition}
\textup{Let $G$ and $H$ be two finite groups and $f: G \rightarrow H$. The {\it left derivative} of $f$ in direction $\alpha \in G$ is defined as the map
\begin{equation}
\begin{array}{l l l l}
d^{(l)}_{\alpha}f: & G & \rightarrow & H\\
& x & \mapsto & f(\alpha x)f(x)^{-1}.
\end{array}
\end{equation}
Symmetrically, the {\it right derivative} of $f$ in direction $\alpha \in G$ is the map
\begin{equation}
\begin{array}{l l l l}
d_{\alpha}^{(r)}f: & G & \rightarrow & H\\
& x & \mapsto & f(x)^{-1} f(x \alpha).
\end{array}
\end{equation}}
\end{definition}
The left-translation actions of both $G$ and $H$ are each equivalent to right-translation actions of $G$ and $H$. Then it is easy to see that each result concerning right-translation action is symmetric to a result on left-translation action. So in this paper, we only focus on the left version and in the remainder we forgot the noun ${\it left}$: the derivative of $f$ in direction $\alpha \in G$, will be simply denoted by $d_{\alpha}f$ and will be the same as $d_{\alpha}^{(l)}f$.

\begin{definition}
\textup{Let $G$ and $H$ be two finite groups and $f: G \rightarrow H$. The map $f$ is said to be {\it perfect nonlinear} if for each $\alpha \in G^*$, $d_{\alpha}f$ is balanced {\it i.e.} for each $(\alpha,\beta) \in G^* \times H$,
\begin{equation}
|\{x \in G | f(\alpha x)f(x)^{-1} = \beta\}| = \displaystyle \frac{|G|}{|H|}.
\end{equation}}
\end{definition}
When $G$ and $H$ are ${\mathbb{Z}}_2^m$ and ${\mathbb{Z}}_2^n$, perfect nonlinearity is a very relevant cryptographic property because it ensures the maximal resistance against the so-called differential cryptanalysis of Biham and Shamir \cite{Biha1}. When $|G| = |H|$, these functions are also known as {\it planar functions} in finite geometry.

\subsection{Bent functions in finite Abelian groups}

When considering the case of finite Abelian groups, it is possible to characterize the notion of perfect nonlinearity using the (discrete) Fourier transform; this {\it dual characterization} leading to an equivalent notion of bent functions. This work has been done recently and independently by Carlet and Ding \cite{CD04} and Pott \cite{Pot04}. This subsection is then devoted to the presentation of these results.\\

For the remainder of this subsection, we suppose given a pair $(G,H)$ of finite {\bf Abelian} groups. The main result obtained by the three authors is essentially based on the following lemma.

\begin{lemma}{\label{equilibre_TF}}(\cite{CD04})
Let $X$ be a finite nonempty set and $f: X \rightarrow H$. The map $f$ is balanced if and only if, for each $\beta \in H^*$, we have
\begin{equation}
\displaystyle \sum_{x \in X} (\chi^{\beta}_H \circ f)(x) = 0.
\end{equation}
\end{lemma}
In particular, if $X$ is a (finite Abelian) group $G$, the previous lemma can be re-written as follows: $f: G \rightarrow H$ is balanced if and only if for each $\beta \in H^*$, $\widehat{(\chi_H^{\beta} \circ f)}(e_G) = 0$. This technical result gives a link between balancedness and the Fourier transform which is used to prove the main result given below.
\begin{theorem}\label{bent_4_Gab_Hab}(\cite{CD04})
Let $f: G \rightarrow H$. The map $f$ is perfect nonlinear if and only if for each $\beta \in H^*$, we have $\forall \alpha \in G$,
\begin{equation}\label{bent_Abelian}
\displaystyle |\widehat{(\chi^{\beta}_H \circ f)}(\alpha)| = \sqrt{|G|}.
\end{equation}
\end{theorem}
When $G$ and $H$ are ${\mathbb{Z}}_2^m$ and ${\mathbb{Z}}_2^n$, a function $f: {\mathbb{Z}}_2^m \rightarrow {\mathbb{Z}}_2^n$ that satisfies the equalities (\ref{bent_Abelian}) is called a {\it boolean bent function}. As perfect nonlinearity, this notion is very important in cryptography because it characterizes the boolean functions that exhibit the best resistance against the linear cryptanalysis of Matsui \cite{Mat94}. By analogy with the boolean case, we will say that a function $f: G \rightarrow H$ that satisfies (\ref{bent_Abelian}) is an {\it (Abelian) bent function}. The theorem above means that Abelian bentness is strictly equivalent to perfect nonlinearity. In the remainder of this paper, we establish the same kind of dual characterization in the cases where $G$ and/or $H$ can be non Abelian.
 
\section{Bent functions in finite non Abelian groups}\label{sect4}

\subsection{Case where $G$ is non Abelian and $H$ is Abelian}
In this subsection, $G$ is a finite {\bf non Abelian} group and $H$ is a finite Abelian group. We first generalize lemma \ref{equilibre_TF} in this context where a non Abelian group occurs.
\begin{lemma}\label{f_pnl_TF_nulle_cas_representations_de_groupe}
Let $f: G \rightarrow H$ and $\rho_0 \in \tilde{G}$ the principal irreducible representation of $G$. The map $f$ is balanced if and only if for each $\beta \in H^*$,
\begin{equation}
\widetilde{(\chi^{\beta}_H \circ f)}(\rho_0) = 0_{{\mathit{End}}({\mathbb{C}})}.
\end{equation}
\end{lemma}

\begin{proof}
First let compute the representation-based Fourier transform of the function $\chi^{\beta}_H \circ f: G \rightarrow {\mathbb{C}}$ at $\rho_0$.
\begin{equation}\label{TF_de_chi_beta_circ_f_cas_representations_lineaires}
\begin{array}{l l l}
\widetilde{(\chi_H^{\beta} \circ f)}(\rho_0) & = & \displaystyle \sum_{x \in G}\chi^{\beta}_H (f(x)) \rho_0 (x)\\
& = & \displaystyle \sum_{x \in G}\chi^{\beta}_H (f(x)) {\mathit{Id}}_{{\mathbb{C}}}\\
& = & \displaystyle \sum_{\gamma \in H}\phi_f (\gamma) \chi^{\beta}_H (\gamma) {\mathit{Id}}_{{\mathbb{C}}}\\
&=& \widehat{\phi_f}(\beta) {\mathit{Id}}_{{\mathbb{C}}}
\end{array}
\end{equation}
where we recall that $\phi_f (\gamma)$ is defined as $|\{x \in G | f(x) = \gamma\}|$.
\begin{itemize}
\item [$\Rightarrow$)] Let $\beta \in H^*$ and suppose that $f$ is balanced. Then $\forall \gamma \in H$, $\phi_f (\gamma) = \displaystyle \frac{|G|}{|H|}$ (by definition of balancedness). According to (\ref{TF_de_chi_beta_circ_f_cas_representations_lineaires}), we find $\displaystyle \widetilde{(\chi_H^{\beta} \circ f)}(\rho_0) = \frac{|G|}{|H|}\sum_{\gamma \in H}\chi^{\beta}_H (\gamma) {\mathit{Id}}_{{\mathbb{C}}}$. Since for each $\beta \in H^*$, we have $\displaystyle \sum_{\gamma \in H}\chi^{\beta}_{H}(\gamma) = 0$ (by lemma \ref{lemme_somme_caractere}), we have $\displaystyle \widetilde{(\chi^{\beta}_H \circ f)}(\rho_0) = 0_{{\mathit{End}}({\mathbb{C}})}$.
\item [$\Leftarrow$)] Let suppose that for each $\beta \in H^*$, $\displaystyle \widetilde{(\chi^{\beta}_H \circ f)}(\rho_0) = 0_{{\mathit{End}}({\mathbb{C}})}$. According to (\ref{TF_de_chi_beta_circ_f_cas_representations_lineaires}), we have for each $\beta \in H^*$, $\widehat{\phi_f}(\beta) {\mathit{Id}}_{{\mathbb{C}}} = 0_{{\mathit{End}}({\mathbb{C}})}$ and then for each $\beta \in H^*$, $\widehat{\phi_f}(\beta) = 0$. The fourth point of proposition \ref{prop_TF_discrete} implies then that $\phi_f$ is constant. Then using the inversion formula, we obtain that for all $\beta \in H$, $\phi_f (\beta) = \displaystyle \frac{1}{|H|}\widehat{\phi_f}(e_H) = \frac{1}{|H|}\sum_{\gamma \in H} \phi_f (\gamma) = \frac{|G|}{|H|}$ (by definition of $\phi_f$). Then $f$ is balanced.
\end{itemize}\qed
\end{proof}
As in the Abelian case, the previous lemma is fundamental for the dual characterization of perfect nonlinearity. Nevertheless before using it, we need an intermediary result.
\begin{proposition}\label{autocorel_G_nab_H_ab}
Let $f: G \rightarrow H$ and $\beta \in H$. We define the {\it autocorrelation function} of $f$ by 
\begin{equation*}
\begin{array}{l l l l}
{\mathit{AC}}_{f,\beta}: & G & \rightarrow & {\mathbb{C}}\\
& \alpha & \mapsto &  (\widetilde{(\chi^{\beta}_H \circ d_{\alpha}f)} (\rho_0))(1).
\end{array}
\end{equation*}
Then for all $\rho = \rho_V \in \widetilde{G}$, 
\begin{equation}
\widetilde{{\mathit{AC}}_{f,\beta}}(\rho) = (\widetilde{(\chi^{\beta}_H \circ f)} (\rho))\circ (\widetilde{(\chi^{\beta}_H \circ f)}(\rho))^{*}.
\end{equation}
\end{proposition}

\begin{proof}
Let $\rho = \rho_V \in \widetilde{G}$.
\begin{equation*}
\begin{array}{l}
\widetilde{{\mathit{AC}}_{f,\beta}} (\rho) =  \displaystyle \sum_{\alpha \in G} {\mathit{AC}}_{f,\beta} (\alpha) \rho(\alpha)\\
= \displaystyle \sum_{\alpha \in G} (\widetilde{(\chi^{\beta}_H \circ d_{\alpha}f)} (\rho_0))(1) \rho (\alpha)\\
 =  \displaystyle \sum_{\alpha \in G}\sum_{x \in G} \chi^{\beta}_H \circ d_{\alpha}f(x) \rho_0 (x)(1) \rho (\alpha)\\
 =  \displaystyle \sum_{\alpha \in G}\sum_{x \in G} \chi^{\beta}_H (f(\alpha x)f(x)^{-1}) \rho (\alpha)\\
\mbox{(by definition of $\rho_0$)}\\
 =  \displaystyle \sum_{\alpha \in G}\sum_{x \in G} \chi^{\beta}_H (f(\alpha x)) \overline{\chi^{\beta}_H (f(x))} \rho (\alpha)\\
 =  \displaystyle \sum_{\alpha \in G}\sum_{x \in G} \chi^{\beta}_H (f(\alpha x)) \overline{\chi^{\beta}_H (f(x))} \rho (\alpha x x^{-1})\\
 =  \displaystyle \sum_{\alpha \in G}\sum_{x \in G} \chi^{\beta}_H (f(\alpha x)) \overline{\chi^{\beta}_H (f(x))} \rho (\alpha x) \circ \rho (x^{-1})\\
\mbox{($\rho$ is a morphism)}\\
 = \displaystyle \sum_{\alpha \in G}\sum_{x \in G} \chi^{\beta}_H (f(\alpha x)) \overline{\chi^{\beta}_H (f(x))} \rho (\alpha x) \circ \rho (x)^{-1}\\
 = \displaystyle \sum_{\alpha \in G}\sum_{x \in G} \chi^{\beta}_H (f(\alpha x)) \overline{\chi^{\beta}_H (f(x))} \rho (\alpha x) \circ \rho(x)^{*}\\
\mbox{($\rho$ is unitary)}\\
 = \displaystyle \sum_{x \in G} \left(\sum_{\alpha \in G} \chi^{\beta}_H (f(\alpha x))  \rho (\alpha x)\right) \circ (\overline{\chi^{\beta}_H (f(x))}\rho (x)^{*})\\
\mbox{(by linearity)}\\
 = \displaystyle \sum_{x \in G}  (\widetilde{(\chi^{\beta}_H \circ f)} (\rho)) \circ (\overline{\chi^{\beta}_H (f(x))}\rho(x)^{*})\\
 = (\widetilde{(\chi^{\beta}_H \circ f)} (\rho)) \circ \left(\displaystyle \sum_{x \in G}   \overline{\chi^{\beta}_H (f(x))}\rho (x)^{*}\right)\\
 = \widetilde{(\chi^{\beta}_H \circ f)} (\rho)) \circ (\widetilde{(\chi^{\beta}_H \circ f)}(\rho))^{*}.
\end{array}
\end{equation*}\qed
\end{proof}

The dual characterization of perfect nonlinearity in this context is given below. This result generalizes the one of Carlet, Ding and Pott.
\begin{theorem}
Let $f: G \rightarrow H$. The map $f$ is perfect nonlinear if and only if $\forall \rho = \rho_V \in \widetilde{G}$ and $\forall \beta \in H^*$, we have
\begin{equation}\label{bent_Gnab_Hab}
(\widetilde{(\chi_H^{\beta} \circ f)}(\rho)) \circ (\widetilde{(\chi_H^{\beta} \circ f)}(\rho))^* = |G|{\mathit{Id}}_V.
\end{equation}
\end{theorem}

\begin{proof}
The map $f$ is perfect nonlinear\\
$\Leftrightarrow$ $\forall \alpha \in G^*,\ d_{\alpha}f$ is balanced (by definition)\\
$\Leftrightarrow$ $\forall \alpha \in G^*,\ \forall \beta \in H^*,\ \widetilde{(\chi^{\beta}_H \circ d_{\alpha}f)} (\rho_0) = 0_{{\mathit{End}}({\mathbb{C}})}$ (according to lemma \ref{f_pnl_TF_nulle_cas_representations_de_groupe})\\
$\Leftrightarrow$ $\forall z \in {\mathbb{C}}$, $\forall \alpha \in G^*,\ \forall \beta \in H^*,\ (\widetilde{(\chi^{\beta}_H \circ d_{\alpha}f)}(\rho_0))(z) = 0$\\
$\Leftrightarrow$ $\forall z \in {\mathbb{C}}$, $\forall \alpha \in G^*,\ \forall \beta \in H^*,\ z{\mathit{AC}}_{f,\beta}(\alpha) = 0$ (by definition of ${\mathit{AC}}_{f,\beta}$)\\
$\Leftrightarrow$ $\forall \alpha \in G^*,\ \forall \beta \in H^*,\ {\mathit{AC}}_{f,\beta}(\alpha) = 0$\\
$\Leftrightarrow$ $\forall \rho = \rho_V \in \widetilde{G},\ \forall \beta \in H^*$, $\widetilde{{\mathit{AC}}_{f,\beta}} (\rho) = {\mathit{AC}}_{f,\beta}(e_G) {\mathit{Id}}_{V}$ (according to the first point of lemma \ref{lemma_diese}).\\
We have 
\begin{equation}
\begin{array}{l l l}
{\mathit{AC}}_{f,\beta} (e_G) &=& (\widetilde{(\chi^{\beta}_H \circ d_{e_G}f)} (\rho_0))(1)\\
&=& \displaystyle \sum_{x \in G} \chi^{\beta}_H (e_H) \rho_0 (x)(1)\\
&=& \displaystyle \sum_{x \in G}\chi^{\beta}_H (e_H)\\
&=& |G|.
\end{array}
\end{equation}
Then $f$ is perfect nonlinear $\Leftrightarrow$ $\forall \beta \in H^*$, $\forall \rho = \rho_V \in \widetilde{G}$, $\widetilde{{\mathit{AC}}_{f,\beta}} (\rho) = |G| {\mathit{Id}}_{V}$\\
$\Leftrightarrow$ $\forall \beta \in H^*$, $\forall \rho \in \widetilde{G}$, $(\widetilde{(\chi^{\beta}_H \circ f)}(\rho))\circ (\widetilde{(\chi^{\beta}_H \circ f)} (\rho))^{*} = |G| {\mathit{Id}}_{V}$ (according to proposition \ref{autocorel_G_nab_H_ab}). \qed
\end{proof}
The functions that satisfy formula (\ref{bent_Gnab_Hab}) are the bent functions in this particular context where $G$ is a finite non Abelian group and $H$ is a finite Abelian group.\\
We can note that this version of bentness is very similar to the one given in theorem \ref{bent_4_Gab_Hab}: the discrete Fourier transform is replaced by its representation-based version, the complex-conjugate is replaced by the adjoint of endomorphisms, the multiplication of complex numbers by the composition of operators and the factor ${\mathit{Id}}_V$ is added. The discrete Fourier transform of Carlet, Ding and Pott's bent functions is, up to a factor $|G|$, ${\mathbb{U}}({\mathbb{C}})$-valued. Regarding this last notion of bentness, the representation-based Fourier transform, also up to the factor $|G|$, is now ${\mathbb{U}}(V)$-valued. It is possible to deduce from this theorem a result really similar to the traditional notion of bentness.
\begin{corollary}
Let $f: G \rightarrow H$. If the map $f$ is perfect nonlinear then we have $\forall \rho = \rho_V \in \widetilde{G}$ and $\forall \beta \in H^*$,
\begin{equation}
\parallel \widetilde{(\chi^{\beta}_H \circ f)}(\rho) \parallel^{2}_{{\mathit{End}}(V)} = |G|\dim_{{\mathbb{C}}}(V),
\end{equation}
where $\parallel \lambda \parallel^2_{{\mathit{End}}(V)} = {\mathit{tr}}(\lambda \circ \lambda^*)$ for $\lambda \in {\mathit{End}}(V)$.
\end{corollary}

\begin{proof}
The result is obvious by using on each member of (\ref{bent_Gnab_Hab}) the trace ${\mathit{tr}}$ of endomorphisms of $V$. \qed
\end{proof}
An interesting question, kept open in this paper, is to know if, whether or not, the reciprocal assertion of the previous corollary is true.

\subsection{Case where $G$ is Abelian and $H$ is non Abelian}

In this subsection, $G$ is a finite Abelian group and $H$ is a finite {\bf non Abelian} group. Another time a technical result similar to both lemmas \ref{equilibre_TF} and \ref{f_pnl_TF_nulle_cas_representations_de_groupe} is needed to establish a dual characterization of perfect nonlinearity in this context.

\begin{lemma}\label{lemmaGAb_HnAb}
Let $X$ be a finite nonempty set and $f: X \rightarrow H$. Then $f$ is balanced if and only if for each $\rho = \rho_V \in \widetilde{H}^*$, we have
\begin{equation}
\displaystyle \sum_{x \in X} (\rho \circ f)(x) = 0_{{\mathit{End}}(V)}.
\end{equation}
\end{lemma}

\begin{proof}
Let $\rho = \rho_V \in \widetilde{H}$. We have 
\begin{equation}
\begin{array}{l l l}
\displaystyle \sum_{x \in X} (\rho \circ f)(x) &=& \displaystyle \sum_{\gamma \in H} |\{x \in X | f(x) = \gamma\}|\rho(\gamma)\\
&=& \displaystyle \sum_{\gamma \in H}\phi_f (\gamma)\rho(\gamma)\\
&=& \widetilde{\phi_f} (\rho).
\end{array}
\end{equation}
\begin{itemize}
\item [$\Rightarrow$)] Let suppose that $f$ is balanced and let $\rho \in \widetilde{H}^*$, then we have 
\begin{equation*}
\displaystyle \sum_{x \in X}(\rho \circ f)(x) = \frac{|X|}{|H|} \sum_{\gamma \in H}\rho(\gamma) = 0_{{\mathit{End}}(V)}
\end{equation*}
(according to lemma \ref{lemma_star}).
\item [$\Leftarrow$)] Let suppose that for all $\rho = \rho_V \in \widetilde{H}^*$, $\displaystyle \sum_{x \in X}(\rho \circ f)(x) = 0_{{\mathit{End}}(V)}$. Then the representation-based Fourier transform of $\phi_f: H \rightarrow \mathbb{N} \subset {\mathbb{C}}$ is
\begin{equation}
\rho_V \mapsto \left \{
\begin{array}{l l}
0_{{\mathit{End}}(V)} & \mbox{if}\ \rho_V \in \widetilde{H}^*,\\
|X| & \mbox{if}\ \rho_V = \rho_0.
\end{array} \right .
\end{equation}
According to lemma \ref{lemma_diese}, we know that $\phi_f$ is constant and more precisely (according to the proof of the lemma), $\forall \beta \in H$, $\phi_f (\beta) = \displaystyle \frac{1}{|X|}{\mathit{tr}}(\widetilde{\phi_f}({\mathit{Id}}_{{\mathbb{C}}}))$. But $\widetilde{\phi_f} ({\mathit{Id}}_{{\mathbb{C}}}) = \displaystyle \sum_{\gamma \in H}\phi_f(\gamma){\mathit{Id}}_{{\mathbb{C}}} = |X| {\mathit{Id}}_{{\mathbb{C}}}$ (by definition of $\phi_f$). Then $\forall \beta \in H$, $\phi_f (\beta) = \displaystyle \frac{|X|}{|H|}$ and $f$ is balanced.
\end{itemize} \qed
\end{proof}
As in the previous case, we introduce a kind of autocorrelation function and we compute its discrete Fourier transform.
\begin{proposition}\label{autocorel_G_ab_H_nab}
Let $f: G \rightarrow H$ and $\rho = \rho_V \in \widetilde{H}$. We define the {\it autocorrelation function} of $f$ by 
\begin{equation*}
\begin{array}{l l l l}
{\mathit{AC}}_{f,\rho}: & G & \rightarrow & {\mathit{End}}(V)\\
& \alpha & \mapsto & \displaystyle \sum_{x \in G} (\rho \circ d_{\alpha}f)(x).
\end{array}
\end{equation*}
Then for all $\alpha \in G$, 
\begin{equation}
\widehat{{\mathit{AC}}_{f,\rho}}^{\mathit{MD}}(\alpha) = (\widehat{(\rho \circ f)}^{\mathit{MD}} (\alpha))\circ (\widehat{(\rho \circ f)}^{\mathit{MD}} (\alpha))^{*}.
\end{equation}
\end{proposition}

\begin{proof}
Let $\alpha \in G$.
\begin{equation}\label{egalite1}
\begin{array}{l}
\widehat{{\mathit{AC}}_{f,\rho}}^{\mathit{MD}}(\alpha) = \displaystyle \sum_{x \in G}\chi^{\alpha}_G (x) {\mathit{AC}}_{f,\rho}(x)\\
= \displaystyle \sum_{x \in G} \chi^{\alpha}_G (x) \sum_{y \in G} (\rho \circ d_{x}f)(y)\\
= \displaystyle \sum_{x \in G} \sum_{y \in G} \chi^{\alpha}_G (x) \rho(f(x y)) \circ (\rho(f(y)))^*\\
\mbox{(since $\rho(x)$ is unitary)}\\
= \displaystyle \sum_{x \in G} \sum_{y \in G} \chi^{\alpha}_G (x y y^{-1}) \rho(f(x y)) \circ (\rho(f(y)))^*\\
= \displaystyle \sum_{x \in G} \sum_{y \in G} \chi^{\alpha}_G (x y) \rho(f(x y)) \circ \overline{\chi^{\alpha}_G (y)}(\rho(f(y)))^*\\
= \displaystyle \sum_{y \in G} \widehat{\rho \circ f}^{\mathit{MD}}(\alpha) \circ (\chi^{\alpha}_G (y) \rho(f(y)))^*\\
= (\widehat{(\rho \circ f)}^{\mathit{MD}}(\alpha)) \circ (\widehat{(\rho \circ f)}^{\mathit{MD}}(\alpha))^*.
\end{array}
\end{equation} \qed
\end{proof}
The corresponding notion of bentness in this context is given by the following theorem.
\begin{theorem}
Let $f: G \rightarrow H$. The map $f$ is perfect nonlinear if and only if $\forall \alpha \in G$, $\forall \rho = \rho_V \in \widetilde{H}^*$,
\begin{equation}\label{bent_Gab_Hnab}
(\widehat{(\rho \circ f)}^{\mathit{MD}}(\alpha)) \circ (\widehat{(\rho \circ f)}^{\mathit{MD}}(\alpha))^* = |G|{\mathit{Id}}_V.
\end{equation}
\end{theorem} 

\begin{proof}

\begin{equation}\label{egalite2}
\begin{array}{l}
$f$\ \mbox{is perfect nonlinear} \Leftrightarrow \forall \alpha \in G^*,\ d_{\alpha}f\ \mbox{is balanced}\\
\Leftrightarrow \forall \alpha \in G^*,\ \forall \rho = \rho_V \in \widetilde{H}^*,\ \displaystyle \sum_{x \in G}(\rho \circ d_{\alpha}f)(x) = 0_{{\mathit{End}}(V)}\\
\mbox{(according to lemma \ref{lemmaGAb_HnAb})}\\
\Leftrightarrow \forall \alpha \in G^*,\ \forall \rho \in \widetilde{H}^*,\ {\mathit{AC}}_{f,\rho}(\alpha) = 0_{{\mathit{End}}(V)}\\
\mbox{(by definition of ${\mathit{AC}}_{f,\rho}$)}\\
\Leftrightarrow \forall \alpha \in G,\ \forall \rho \in \widetilde{H}^*,\ \widehat{{AC}_{f,\rho}}^{\mathit{MD}}(\alpha) = {\mathit{AC}}_{f,\rho}(e_G)\\
 \mbox{(according to proposition \ref{prop_TF_multidimensional}).}
\end{array}
\end{equation}
But ${\mathit{AC}}_{f,\rho}(e_G) = \displaystyle \sum_{x \in G}(\rho \circ d_{e_G}f)(x) = \sum_{x \in G}\rho(e_H) = \sum_{x \in G}{\mathit{Id}}_V = |G| {\mathit{Id}}_V$. Then according to (\ref{egalite1}) and (\ref{egalite2}),
$f$ is perfect nonlinear  $\Leftrightarrow$ $\forall \alpha \in G,\ \forall \rho = \rho_V \in \widetilde{H}^*$,
\begin{equation} 
(\widehat{(\rho \circ f)}^{\mathit{MD}} (\alpha)) \circ (\widehat{(\rho \circ f)}^{\mathit{MD}} (\alpha))^* = |G| {\mathit{Id}}_V.
\end{equation} \qed
\end{proof}
Another time, by using the trace on both sides of (\ref{bent_Gab_Hnab}), we deduce the following corollary. As in the previous case, an interesting question should be to check if this result is or not a sufficient condition for bentness in this particular context.
\begin{corollary}
Let $f: G \rightarrow H$. If the map $f$ is perfect nonlinear then $\forall \alpha \in G$ and $\forall \rho = \rho_V \in \widetilde{H}^*$,
\begin{equation}
\parallel \widehat{(\rho \circ f)}^{\mathit{MD}}(\alpha) \parallel_{{\mathit{End}}(V)}^2 = |G|\dim_{{\mathbb{C}}}(V).
\end{equation}
\end{corollary}

\subsection{Case where $G$ and $H$ are both non Abelian}

In this subsection, $G$ and $H$ are both finite {\bf non Abelian} groups.\\

Let $\rho' = \rho_{W}' \in \widetilde{H}$ and $\displaystyle B = \{e_i\}_{i=1}^{\dim_{{\mathbb{C}}}(W)}$ be an orthonormal basis of $W$ (for the scalar product $\langle .,. \rangle_{H,\rho',W}$ of $W$ as introduced by (\ref{inner_product_inv})) in which for all $y \in H$, $\rho'(y)$ is a unitary operator. For $(i,j) \in \{1,\ldots,\dim_{{\mathbb{C}}}(W)\}^2$, let define 
\begin{equation}
\begin{array}{l l l l}
\rho_{ij}': & H & \rightarrow & {\mathbb{C}}\\
& y & \mapsto & \langle \rho'(e_i),e_j\rangle_{H,\rho',W}.
\end{array}
\end{equation} 
In other terms, for each $y \in H$, $\rho_{ij}'(y)$ is simply the coefficient $(i,j)$ of the $\dim_{{\mathbb{C}}}(W) \times \dim_{{\mathbb{C}}}(W)$ unitary matrix that represents $\rho'(y)$ in the basis $B$.\\
Let see some obvious results on $\rho_{ij}'$ for $(i,j) \in \{1,\ldots,\dim_{{\mathbb{C}}}(W)\}^2$.
\begin{enumerate}
\item Let $(y_1,y_2) \in H^2$. We have $\rho'(y_1 y_2) = \rho' (y_1) \circ \rho'(y_2)$. Then we have $\rho_{ij}'(y_1 y_2) = \displaystyle \sum_{k = 1}^{\dim_{{\mathbb{C}}}(W)} \rho_{ik}' (y_1)\rho_{kj}' (y_2)$;
\item Let $y \in H$. Since $\rho'(y^{-1}) = \rho'(y)^*$ then we deduce that $\rho_{ij}'(y^{-1}) = \overline{\rho_{ji}'(y)}$.
\end{enumerate}
Note also that the identity map ${\mathit{Id}}_W$ is written in any orthonormal basis of $W$ as the identity matrix and $0_{{\mathit{End}}(W)}$ is associated, in any basis of $W$, with the all-zero matrix.\\

As in the previous subsections, we introduce some kind of autocorrelation function for $f: G \rightarrow H$. 

\begin{proposition}
Let $f: G \rightarrow H$, $\rho' = \rho_W' \in \widetilde{H}$ and $(i,j) \in \{1,\ldots,\dim_{{\mathbb{C}}}(W)\}^2$. We define the {\it autocorrelation function} of $f$
\begin{equation}
\begin{array}{l l l l}
{\mathit{AC}}_{f,\rho',i,j}: & G & \rightarrow & {\mathbb{C}}\\
& \alpha & \mapsto & \displaystyle \sum_{x \in G}(\rho_{ij}' \circ d_{\alpha}f)(x).
\end{array}
\end{equation}
Then for all $\rho = \rho_V \in \widetilde{G}$,
\begin{equation}
\widetilde{{\mathit{AC}}_{f,\rho',i,j}}(\rho) = \displaystyle \sum_{k=1}^{\mathit{dim}_{{\mathbb{C}}}(W)}(\widetilde{(\rho_{ik}' \circ f)}(\rho)) \circ (\widetilde{(\rho_{jk}' \circ f)}(\rho))^*.
\end{equation}
\end{proposition}

\begin{proof}
Let $\rho = \rho_V \in \widetilde{G}$.
\begin{equation}\label{egalite3}
\begin{array}{l}
\widetilde{{\mathit{AC}}_{f,\rho',i,j}}(\rho) = \displaystyle \sum_{x \in G} {\mathit{AC}}_{f,\rho',i,j}(x) \rho(x)\\
= \displaystyle \sum_{x \in G} \sum_{y \in G} (\rho_{ij}' \circ d_x f)(y) \rho(x)\\
=\displaystyle \sum_{x \in G} \sum_{y \in G} \rho_{ij}' (f(xy)f(y)^{-1}) \rho(x)\\
= \displaystyle \sum_{x \in G} \sum_{y \in G} \sum_{k = 1}^{\dim_{{\mathbb{C}}}(W)} \rho_{ik}'(f(xy))\overline{\rho_{jk}'(f(y))} \rho(x)\\
= \displaystyle \sum_{k=1}^{\dim_{{\mathbb{C}}}(W)} \sum_{x \in G} \sum_{y \in G}  \rho_{ik}'(f(xy))\overline{\rho_{jk}'(f(y))} \rho(xyy^{-1})\\
= \displaystyle \sum_{k=1}^{\dim_{{\mathbb{C}}}(W)} \sum_{x \in G} \sum_{y \in G}  \rho_{ik}'(f(xy))\overline{\rho_{jk}'(f(y))} \rho(xy) \circ \rho(y)^*\\
= \displaystyle \sum_{k = 1}^{\dim_{{\mathbb{C}}}(W)} \sum_{x \in G} \sum_{y \in G} \rho_{ik}'(f(xy)) \rho(xy) \circ \left ( \overline{\rho_{jk}'(f(y))} \rho(y)^*\right)\\
= \displaystyle \sum_{k = 1}^{\dim_{{\mathbb{C}}}(W)}  (\widetilde{(\rho_{ik}' \circ f)}(\rho)) \circ \sum_{y \in G} \left ( \rho_{jk}'(f(y)) \rho(y)\right )^*\\
=\displaystyle \sum_{k = 1}^{\dim_{{\mathbb{C}}}(W)} (\widetilde{(\rho_{ik}' \circ f)}(\rho)) \circ (\widetilde{(\rho_{jk}' \circ f)}(\rho))^*.
\end{array}
\end{equation} \qed
\end{proof}
Using this autocorrelation function and its Fourier transform we can exhibit the appropriate notion of bentness for this context where both groups $G$ and $H$ are non Abelian.
\begin{theorem}
Let $f: G \rightarrow H$. The map $f$ is perfect nonlinear if and only if $\forall \rho = \rho_V \in \widetilde{G}$, $\forall \rho' = \rho_{W}' \in \widetilde{H}^*$, $\forall (i,j) \in \{1,\ldots,\dim_{{\mathbb{C}}}(W)\}^2$,
\begin{equation}\label{BentGHNAb}
\displaystyle \sum_{k=1}^{\dim_{{\mathbb{C}}}(W)} (\widetilde{(\rho_{ik}' \circ f)}(\rho)) \circ (\widetilde{(\rho_{jk}' \circ f)}(\rho))^* = \left \{
\begin{array}{l}
|G| {\mathit{Id}}_V\ \mbox{if}\ i=j,\\
0_{{\mathit{End}}(V)}\ \mbox{if}\ i \not = j.
\end{array} \right .
\end{equation}
\end{theorem}

\begin{proof}
\begin{equation}\label{egalite4}
\begin{array}{l}
f\ \mbox{is perfect nonlinear} \Leftrightarrow \forall \alpha \in G^*,\ d_{\alpha}f\ \mbox{is balanced}\\
\Leftrightarrow \forall \alpha \in G^*,\ \forall \rho' = \rho_{W}' \in \widetilde{H}^*,\\
 \quad \quad \quad \displaystyle \sum_{x \in G}(\rho' \circ d_{\alpha}f)(x) = 0_{{\mathit{End}}(W)}\\
\mbox{(according to lemma \ref{lemmaGAb_HnAb})}\\
\Leftrightarrow \forall \alpha \in G^*,\ \forall \rho' \in \widetilde{H}^*,\ \forall (i,j) \in \{1,\ldots,\dim_{{\mathbb{C}}}(W)\}^2,\\
 \quad \quad \quad\displaystyle \sum_{x \in G}(\rho_{ij}' \circ d_{\alpha}f)(x) = 0\\
\Leftrightarrow \forall \alpha \in G^*,\ \forall \rho' \in \widetilde{H}^*,\ \forall (i,j) \in \{1,\ldots,\dim_{{\mathbb{C}}}(W)\}^2,\\
 \quad \quad \quad{\mathit{AC}}_{f,\rho',i,j}(\alpha) = 0\\
\Leftrightarrow\forall \rho_V \in \widetilde{G},\ \forall \rho' \in \widetilde{H}^*,\ \forall (i,j) \in \{1,\ldots,\dim_{{\mathbb{C}}}(W)\}^2,\\
 \quad \quad \quad\widetilde{{\mathit{AC}}_{f,\rho',i,j}}(\rho_V) = {\mathit{AC}}_{f,\rho',i,j}(e_G){\mathit{Id}}_V\\
\mbox{(by lemma \ref{lemma_diese}).}
\end{array}
\end{equation}
But ${\mathit{AC}}_{f,\rho',i,j}(e_G) = \displaystyle \sum_{x \in G}(\rho_{ij}' \circ d_{e_G}f)(x) = \sum_{x \in G}\rho_{ij}'(e_H)$. Since we know that $\rho'(e_H) = {\mathit{Id}}_W$, then $\rho'(e_H)$ is written in the orthonormal basis $B$ of $W$ as the identity matrix and then $\forall (i,j) \in \{1,\ldots,\dim_{{\mathbb{C}}}(W)\}^2$,
\begin{equation}
\rho_{ij}'(e_H) = \left \{
\begin{array}{l l}
1 & \mbox{if}\ i=j,\\
0 & \mbox{if}\ i \not = j.
\end{array}
\right .
\end{equation}
From this last result, the equality (\ref{egalite3}) and the equivalence (\ref{egalite4}), it follows the expected result. \qed
\end{proof}
This case where both groups $G$ and $H$ are non Abelian involves some kind of tensor (or at least of block-matrix) notion of bentness. This is essentially due to the lack of commutativity of both groups.

\section{Summary}

The different notions of bentness, depending on the fact that the finite groups $G$ and $H$ are Abelian or not, are summarized below.\\
A function $f: G \rightarrow H$ is {\it bent} (or equivalently perfect nonlinear) if and only if
\begin{enumerate}
\item If $G$ and $H$ are Abelian (\cite{CD04,Pot04}):  $\forall (\alpha,\beta) \in G \times H^*$,
\begin{equation*}
|\widehat{(\chi^{\beta}_G \circ f)}(\alpha)|^2 = |G|.
\end{equation*}
\item If $G$ is non Abelian and $H$ is Abelian: $\forall (\rho,\beta) \in \tilde{G} \times H^*$ (with $\rho: G \rightarrow {\mathbb{U}}(V)$),
\begin{equation*}
(\widetilde{(\chi^{\beta}_H \circ f)}(\rho)) \circ (\widetilde{(\chi^{\beta}_H \circ f)}(\rho))^* = |G| {\mathit{Id}}_V.
\end{equation*}
\item If $G$ is Abelian and $H$ is non Abelian: $\forall (\alpha,\rho') \in G \times \widetilde{H}^*$ (with $\rho': H \rightarrow {\mathbb{U}}(W)$):
\begin{equation*}
(\widehat{(\rho' \circ f)}^{\mathit{MD}}(\alpha)) \circ (\widehat{(\rho' \circ f)}^{\mathit{MD}}(\alpha))^* = |G| {\mathit{Id}}_W.
\end{equation*} 
\item If $G$ and $H$ are both non Abelian groups: $\forall (\rho,\rho',(i,j)) \in \widetilde{G} \times \widetilde{H}^* \times \{1,\ldots,\dim_{{\mathbb{C}}}(W)\}^2$ (with $\rho: G \rightarrow {\mathbb{U}}(V)$ and $\rho': H \rightarrow {\mathbb{U}}(W)$),
\begin{equation*}
\displaystyle \sum_{k = 1}^{\dim_{{\mathbb{C}}}(W)}  (\widetilde{(\rho_{ik}' \circ f)}(\rho)) \circ (\widetilde{(\rho_{jk}' \circ f)}(\rho))^* = \left \{
\begin{array}{cc}
|G|{\mathit{Id}}_V & \mbox{if}\ i=j,\\
0_{{\mathit{End}}(V)} & \mbox{if}\ i\not=j.
\end{array}\right .
\end{equation*}
\end{enumerate}

\end{document}